\documentclass[conference,10pt]{IEEEtran}
\IEEEoverridecommandlockouts
\usepackage{cite}
\usepackage{amsmath,amssymb,amsfonts}
\usepackage{algorithmic}
\usepackage{graphicx}
\usepackage{textcomp}
\usepackage{xcolor}
\def\BibTeX{{\rm B\kern-.05em{\sc i\kern-.025em b}\kern-.08em
    T\kern-.1667em\lower.7ex\hbox{E}\kern-.125emX}}
\usepackage{cuted}
\usepackage{amsthm}

\newtheorem{lemma}{Lemma}
\theoremstyle{remark}

\newtheorem{definition}{Definition}
\usepackage{mathtools}
\usepackage{titlesec}
\titlespacing{\subsection}{0pt}{0.5\baselineskip}{0.1\baselineskip}
\begin{document}

\title{Modeling UAV-aided Roadside Cell-Free Networks with Matérn Hard-Core Point Processes
}

\author{
    \IEEEauthorblockN{Chenrui Qiu\IEEEauthorrefmark{1}, Yongxu Zhu\IEEEauthorrefmark{2}, Bo Tan\IEEEauthorrefmark{3}, George K. Karagiannidis\IEEEauthorrefmark{4}, Tasos Dagiuklas\IEEEauthorrefmark{1}}
    \IEEEauthorblockA{\IEEEauthorrefmark{1} School of Computer Science and Digital Technologies, London South Bank University, London, UK}
    \IEEEauthorblockA{\IEEEauthorrefmark{2} School of Information Science and Engineering, Southeast University, Nanjing, China}
    \IEEEauthorblockA{\IEEEauthorrefmark{3} Department of Electronic and Electrical Engineering, University College London, London, UK}
    \IEEEauthorblockA{\IEEEauthorrefmark{4} Department of Electrical and Computer Engineering, Aristotle University of Thessaloniki, Thessaloniki, Greece}
    \IEEEauthorblockA{\{qiuc3, tdagiuklas\}@lsbu.ac.uk, yongxu.zhu@seu.edu.cn, tan.bo@ucl.ac.uk, geokarag@auth.gr}
}

\maketitle

\begin{abstract}
This paper investigates a uncrewed aerial vehicles (UAV)-assisted cell-free architecture for vehicular networks in road-constrained environments. Roads are modeled using a Poisson Line Process (PLP), with multi-layer roadside access points (APs) deployed via 1-D Poisson Point Process (PPP). Each user forms a localized cell-free cluster by associating with the nearest AP in each layer along its corresponding road. This forms a road-constrained cell-free architecture. To enhance coverage, UAV act as an aerial tier, extending access from 1-D road-constrained layouts (embedded in 2-D) to 3-D. We employ a Matérn Hard-Core (MHC) point process to model the spatial distribution of UAV base stations, ensuring a minimum safety distance between them.  
In order to enable tractable analysis of the aggregate signal from multiple APs, a distance-based power control scheme is introduced.  
Leveraging tools from stochastic geometry, we have studied the coverage probability. Furthermore, we analyze the impact of key system parameters on coverage performance, providing useful insights into the deployment and optimization of UAV-assisted cell-free vehicular networks.

\end{abstract}

\begin{IEEEkeywords}
Uncrewed aerial vehicles, Cell-free, PLP, Matérn Hard-Core process.
\end{IEEEkeywords}

\section{Introduction}


The integration of uncrewed aerial vehicles (UAV) into cell-free massive MIMO architectures has received considerable attention due to its potential to enhance coverage and flexibility in next-generation wireless networks. A number of studies have examined the theoretical foundations of such integration, often employing stochastic geometry to analyze large-scale UAV distributions and their effects on system performance. For instance, \cite{RN172} proposes a stochastic geometry-based framework for evaluating the SINR and outage performance in aerial-user-enabled cell-free networks, while \cite{RN169} investigates cache-assisted UAV networks for emergency scenarios using Poisson Point Process (PPP) modeling. Additionally, \cite{RN181} considers underlay spectrum access with UAV-assisted cell-free networks, highlighting interference mitigation mechanisms. At the system level, \cite{RN183} provides one of the earliest analyses of UAV-assisted cell-free networks with a focus on channel hardening and pilot contamination, and \cite{RN184} explores the asymptotic behavior and optimal deployment strategies of UAV-based cell-free architectures. Survey-oriented works such as \cite{RN177} and \cite{RN175} summarize recent advancements in UAV swarming, mobility, and cell-free system applications, especially for IIoT and smart infrastructure.

Parallel to theoretical modeling, several studies have focused on optimization and deployment strategies for enhancing cell-free UAV performance. 
\cite{RN182} proposes joint user association and power allocation for fairness optimization in UAV-assisted cell-free networks, while \cite{RN179} formulates a weighted sum power minimization problem under UAV energy constraints. 
Beamforming and placement optimization are studied in \cite{RN178}, where an intelligent reflecting surface (IRS) is mounted on UAVs to improve spatial multiplexing. 
In a broader network context, \cite{RN180} examines UAV-assisted satellite and terrestrial integrated cell-free systems with energy efficiency objectives. Mobility and association dynamics are also investigated in \cite{RN174}, which proposes UAV-centric dynamic access point assignment for aerial–to-ground cell-free networks. 
Furthermore, \cite{RN165} explores NOMA-based UAV-assisted systems that could complement cell-free architectures.

Despite the progress in system-level optimization, most existing studies rely on channel or power-based association while overlooking traffic heterogeneity and queueing-induced latency, leading to excessive delay and inefficient resource utilization in dense networks \cite{ji2025delay}. Moreover, UAV-assisted cell-free solutions often adopt idealized assumptions, such as perfect hardware and high-capacity fronthaul, whereas practical imperfections can substantially degrade UAV communication performance \cite{11204827}, calling for more realistic and deployable designs.

Thus, this paper considers a hybrid architecture wherein line-distributed ground APs form a stable cell-free system, and UAV act as supplementary aerial base stations.
However, most existing studies assume uniform or unstructured AP distributions. In contrast, vehicular networks feature road-constrained AP layouts, which are rarely modelled. This motivates the need for a new framework that incorporates line-deployed cell-free APs.

Roads are modeled via a Poisson Line Process (PLP), with APs deployed along each road using 1-D PPPs and UAV forming a 3-D aerial tier. This results in a hierarchical 1-D to 3-D access structure suited to vehicular topologies.

In this work our main contributions are: (1) We develop an integrated spatial model combining PLP-based roadside APs and 3-D Matérn Hard-Core (MHC)-based UAV as base stations in case of collision.
(2) A road-constrained cell-free architecture is proposed, where each user forms a localized cooperation cluster by associating with the nearest AP in each layer along its corresponding road.
(3) We derive expressions for coverage using distance-based power control and Gamma approximations.
(4) Extensive simulations validate the analysis and identify optimal key parameters for UAV deployment in 
roadside cell-free network.

\section{System Model}
We consider a downlink cell-free vehicular network comprising ground APs and UAV that collaboratively serve vehicular users, as illustrated in Fig.~\ref{model}.
The road system is modeled by a motion-invariant PLP $\Phi_l$ in $\mathbb{R}^2$ with line density $\mu_l$, inducing a poisson point process in the representation space $\mathcal{C}$ with density $\lambda_l = \mu_l/\pi$. Each line in $\Phi_l$ corresponds to a straight road.
Along each road, APs are deployed following a homogeneous 1-D PPP with $K$ layers, each with density $\lambda_{\rm a}^{(k)}$. APs are single-antenna and connect with the same edge cloud instance in each road.
UAV act as aerial Base Stations (BSs) and are distributed in 3-D space according to a MHC process $\Phi_{\rm U}$ with density $\lambda_{\rm U}$ and safety distance $d$. All UAV hover at altitude $H_{\rm U}$ and use single-antenna transmission.
Vehicular users are independently distributed along roads via a 1-D PPP and equipped with single receive antennas. All APs and UAV transmit at powers $P_{\rm a}$ and $P_{\rm U}$, respectively. A typical user (UE) is located at the origin $\mathbf{O} = (0,0)$.
\begin{figure}[h]
\centering
\includegraphics[width=6.5cm, height=5cm]{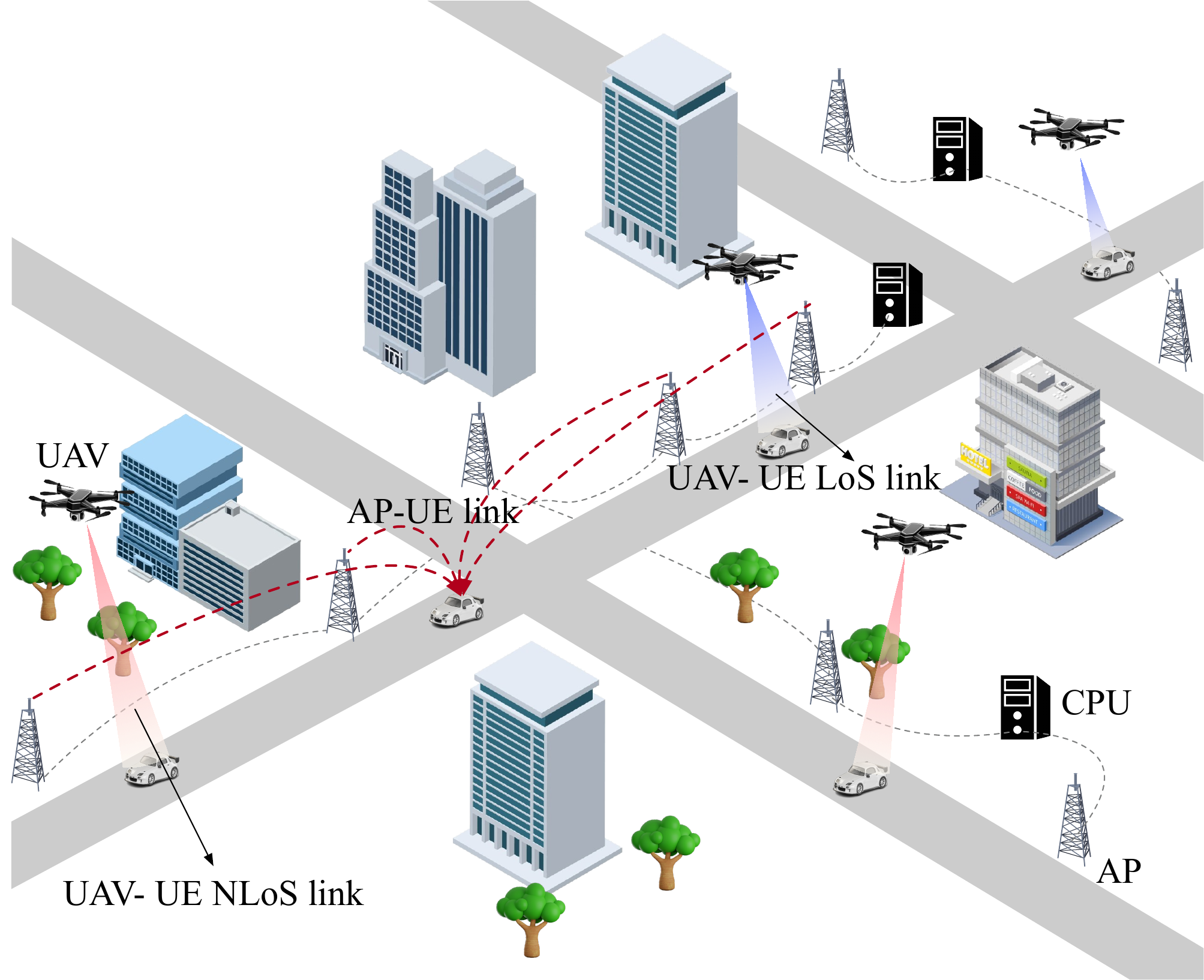} 
\caption{Architecture of vehicular Cell-free Network integrated with UAV base stations.}
\label{model}
\end{figure}

To ensure minimum spacing and mitigate UAV clustering, UAV are modeled via a Type-II MHC point process $\Phi_{\rm U}$ as \cite{article009090}, obtained by thinning a parent PPP $\Phi_p$ with intensity $\lambda_p$. Each point in $\Phi_p$ is assigned a uniform random mark in $[0,1]$ and retained only if it has the lowest mark within a disk of radius $d$. The resulting intensity is
$
\lambda_{\rm U} = \frac{1 - \exp(-\lambda_p \pi d^2)}{\pi d^2}.
$

Due to the spatial dependence in MHC, we employ the second-order product density $\zeta^{(2)}(v)$ to characterize the probability of finding two points at distance $v$:
\begin{equation}
\setlength{\abovedisplayskip}{4pt}
\setlength{\belowdisplayskip}{4pt}
\zeta^{(2)}(v) =
\begin{cases}
0, & v < d \\  
\frac{
2V(v) \left[1 - e^{-\lambda_p \pi d^2} \right] - 2\pi d^2 \left[1 - e^{-\lambda_p V(v)} \right]
}{
\pi d^2 V(v) \left[ V(v) - \pi d^2 \right]
}, & d \leq v \leq 2d  \\
\lambda_u^2, & v > 2d 
\end{cases}
\end{equation}
where $V(v)$ is the union area of two disks separated by $v$:
\begin{align}
V(v) = 2\pi d^2 - 2d^2 \cos^{-1}\left( \tfrac{v}{2d} \right) + v \sqrt{d^2 - \tfrac{v^2}{4}}.
\end{align}

\subsection{Channel Model}
UAV-UE links may be either line-of-sight (LoS) or non-line-of-sight (NLoS). Given UAV altitude $H_{\rm U}$ and horizontal distance $z$, the LoS probability is modeled as \cite{6863654}
\begin{equation}
p_{\mathrm{L}}(z)
= \frac{1}{1 + a \exp\left(-b\left(\tfrac{180}{\pi}\tan^{-1}\left(\tfrac{H_{\rm U}}{z}\right) - a\right)\right)},
\end{equation}
where $a$ and $b$ are environment-dependent parameters. The NLoS probability is $p_{\mathrm{NL}}(z) = 1 - p_{\mathrm{L}}(z)$.

We consider a directional beamforming strategy where antenna main lobes are aligned along the road direction. The ground AP antenna pattern is sectorized, with main and side lobe gains $G_{\rm a}$ and $g_{\rm a}$, respectively. UAV are similarly equipped with directional antennas, characterized by $G_{\rm U}$ and $g_{\rm U}$. The serving antenna always steers its main lobe toward the intended user.

We assume that UAV-UE links experience Nakagami-$m$ fading with parameters $m_{\rm L}$ and $m_{\rm NL}$ for LoS and NLoS, respectively. For AP-UE links, we adopt Nakagami-$m$ fading with parameters $m_{\rm 21}$ and $m_{\rm 22}$, corresponding to typical and non-typical lines. The respective channel gains are denoted by $h_{\rm L}$, $h_{\rm NL}$, $h_{\rm a_0}$, and $h_{\rm a_1}$, with all $m$ values assumed to be positive integers.

Large-scale fading due to path loss is considered, with exponent $\alpha_{\rm a} \geq 2$ for AP-UE links and $\alpha_v$ ($v \in \{\mathrm{L,NL}\}$) for UAV-UE links under LoS and NLoS conditions.
\subsection{Association Policy}
We assume that the typical user associates with the AP or UAV providing the strongest biased received power. To enable load balancing, selection bias factors $B_{\rm U}$ and $B_{\rm a}$ are applied for \textit{tier}-1 (UAV) and \textit{tier}-2 (APs), respectively. The serving node $r^*$ is determined as
\begin{align}
r^* = \arg \max \left\{ B_{\rm U} P_{\rm U} G_{\rm U} r_1^{-\alpha_v},\; B_{\rm a} P_{\rm a} G_{\rm a} r_2^{-\alpha_{\rm a}} \right\}, 
\end{align}
where $r_1$ denotes the 3-D distance between the typical UE and its nearest UAV, and  $r_2 = \arg \max_{r \in \Xi_{L_{0}}} B_{\rm a} P_{\rm a} G_{\rm a} r^{-\alpha_{\rm a}} $, and $v \in \{\mathrm{L, NL}\}$ denotes the LoS or NLoS condition.

Since UAV are modeled as a homogeneous PPP in $\mathbb{R}^2$, the horizontal distance $z$ between the user and its nearest UAV follows a Rayleigh distribution with the PDF
$f_Z(z) = 2\pi \lambda_{\rm U} z \exp(-\pi \lambda_{\rm U} z^2).$
The probability that the typical user associates with a UAV via a LoS link is given by
\begin{align}
\mathcal{A}_{\mathrm{L}} = \mathbb{E}_Z \left[ p_{\mathrm{L}}(z) \right] 
= \int_0^\infty p_{\mathrm{L}}(z) \cdot f_Z(z)\, dz,
\end{align}
and the NLoS association probability is $\mathcal{A}_{\mathrm{NL}} = 1 - \mathcal{A}_{\mathrm{L}}$.

\subsection{SINR and Coverage Probability}
The received signal power at the typical receiver $\textbf{O}$ is decided by the transmit node. When the receiver is associated with tier-1 UAV node, the received power can be expressed by 
\begin{equation}
P_r = P_{\rm U} G_{\rm U} h_{v} r^{-\alpha_v}, \quad v \in \{\mathrm{L}, \mathrm{NL}\},
\end{equation}
where $P_{\rm U}$ is the UAV transmit power, $h_{v}$ is the Nakagami-$m$ fading gain, and $\alpha_v$ represents the path loss exponent under LoS or NLoS conditions.

If the typical receiver is served by the closest AP of each layer on the road, the received power can be defined as
\begin{equation}
\setlength{\abovedisplayskip}{4pt}
\setlength{\belowdisplayskip}{4pt}
P_r = \sum_{k} P_{\rm a} G_{\rm a} h_{{\rm a}_0}^{(k)} r_k^{-\alpha_{\rm a}},
\end{equation}
where $r_k$ and $h_{{\rm a}_0}^{(k)}$ denote the distance and fading gain associated with the $k$-th layer closest AP.

Based on the above assumptions, the SINR received at the typical receiver can be expressed as
$\gamma = \frac{P_r} {\sigma^2 + \mathcal{I}_1+ \mathcal{I}_{2}},$
where $ \mathcal{I}_1$ is the interference from the \textit{tier}-1 nodes,   $ \mathcal{I}_{2}$ is interference from the \textit{tier}-2 nodes.
To evaluate the probability that the typical receiver achieves a rate above a target threshold $\gamma_0$, we define the coverage probability as
$\mathbb {P}_{\rm C}(\gamma\geq\gamma_0)$ \cite{11195162}.


\section{Evaluation of Performance Metrics}
This section analyzes key metrics for the proposed UAV-assisted vehicular network. We first derive the association probabilities under biased received power association. Then, we characterize the serving distance distributions, which form the basis for SINR and coverage analysis.

\subsection{Association Probability}
\begin{definition}
The 3-D distance $r_1$ from the typical user to its serving UAV has the CDF and PDF as
\begin{align}
F_{R_1}(r_1) &= 1 - \exp\left( - \lambda_{\rm U} \pi (r_1^2 - {H_{\rm U}}^2) \right), r_1 \geq {H_{\rm U}}, \\
f_{R_1}(r_1) &= 2 \pi \lambda_{\rm U} r_1 \exp\left( - \lambda_{\rm U} \pi (r_1^2 - {H_{\rm U}}^2) \right), r_1 \geq {H_{\rm U}}. \label{eq:pdf_r1}
\end{align}
\end{definition}

\begin{definition}
The 1-D distance $r_2$ from the user to its nearest AP along the road follows
\begin{align}
F_{R_2}(r_2) &= 1 - \exp\left( -2 \lambda_{\rm a} r_2 \right),  r_2 \geq 0, \\
f_{R_2}(r_2) &= 2 \lambda_{\rm a} \exp\left( -2 \lambda_{\rm a} r_2 \right),  r_2 \geq 0. 
\end{align}
\end{definition}

\begin{lemma}
We denote that typical receiver connects to UAV (\textit{tier}-1) as event \(\mathcal{E}_1\), to APs (\textit{tier}-2) as event \(\mathcal{E}_2\). The probability of the event \(\mathcal{E}_1\) and \(\mathcal{E}_2\) defined as
\begin{equation}
    \begin{aligned}
        \begin{cases} 
        \mathbb{P}(\mathcal{E}_1) = \mathcal{A}_{\mathrm{L}}\mathcal{A}_{\mathrm{L}}^1 + \mathcal{A}_{\mathrm{NL}}\mathcal{A}_{\mathrm{NL}}^1,\\
        \mathbb{P}(\mathcal{E}_2) = 1-\mathbb{P}(\mathcal{E}_1).
        \end{cases}
    \end{aligned} \label{pe1}
\end{equation}
where $\mathcal{A}_{v}^1, v\in \{\rm {L,NL}\}$ means UE biased received power at LoS or NLoS UAV-UE link is larger than from APs, could be expressed as  
\begin{align}\label{user_association}
	\mathcal{A}_{v}^1 &=  \mathbb{P} \left ({P_{\rm U}B_{\rm U}G_{\rm U}R_{1}^{-\alpha_v } > P_{\rm a}B_{\rm a}G_{\rm a} R_{2}^{-\alpha_{\rm a} }}\right) \\&= 1 - \int_0^\infty  {{e^{ - {\lambda_{\rm U}}\pi \rho(r_2)}}}  \times 2{\lambda _a}{e^{ - 2{\lambda _a}{r_2}}}{\mkern 1mu} {\rm{d}}{r_2},\notag
\end{align}
where $\rho(r_2) = \max\left( 0, \xi_{21}^{ - \frac{2}{\alpha_v} } r_2^{ \frac{2 \alpha_{\rm a}}{ \alpha_v } } - H_{\rm U}^2 \right)$.
\end{lemma}
\begin{proof}
    See Appendix ~\ref{appendixB}. 
\end{proof}

\subsection{Serving Distance Distribution}
Due to distinct association scenarios under $\mathcal{E}_1$ and $\mathcal{E}_2$, the corresponding serving distance distributions are evaluated separately to precisely reflect their influence on coverage performance in heterogeneous networks.

\begin{lemma}
\label{lemma2}
Conditioned on the event $\mathcal{E}_1$ and $\mathcal{E}_2$, the PDF of the serving distance is given by
\begin{align}
	&{f_R}(r|{\mathcal{A}_{v}^1 }) = \frac{{2\pi {\lambda _u}r}}{{{\mathbb{P}}(\mathcal{A}_{v}^1 )}}{e^{ - \pi {\lambda _u}({r^2-{H_{\rm U}}^2}) - 2{\lambda _a}\xi _{21}^{1/{\alpha_{\rm a}}}{r^{{\alpha _v}/{\alpha_{\rm a}}}}}},\\
	&{f_R}(r|{\mathcal{A}_{v}^2}) = \frac{{2{\lambda _a}}}{{{\mathbb{P}}(\mathcal{A}_{v}^2)}}{e^{ - 2{\lambda _a}r - \pi {\lambda _u}\rho(r)}}. 
\end{align}
where $\rho (r) = \max (0,\xi _{21}^{\frac{{ - 2}}{{{\alpha _v}}}}{r^{\frac{{2{\alpha _{\rm{a}}}}}{{{\alpha _v}}}}} - {H_{\rm{U}}}^2)$, $\xi_{21} = \frac{P_{\rm a}B_{\rm a}G_{\rm a}}{P_{\rm U}B_{\rm U}G_{\rm U}}$, $v\in \{\rm {L,NL}\}$,
\end{lemma}
\begin{proof}
See Appendix~\ref{appendixC}.
\end{proof}

\subsection{Coverage Probability}
Power control is employed to mitigate path loss and enhance reliability in collaborative downlink transmission. For the $k$-th layer AP at distance $r_{2}$, the transmit power is defined as
\begin{align}
P_2 = C r_{2}^{\alpha_{\rm a}},
\end{align}
where $\alpha_{\rm a}$ is the path loss exponent and $C$ is a power scaling constant.

To ensure tractability, power control is applied only to desired signals from the nearest AP in each layer on the typical line, while interference remains unaffected. This selective strategy decouples the desired signal from the interference field, simplifying SINR analysis under PLP-induced spatial correlation. It also enables tractable characterization of coverage in multi-layer cooperative networks.

Thus, coverage probability is 
\begin{align} \label{pc2}
\mathbb {P}_{\rm C} &=\sum_{v \in\{\mathrm{L}, \mathrm{NL}\}}\mathcal{A}_{v}\mathcal{A}_{v}^1\mathbb{P}(\frac{P_{\rm U} G_{\rm U} h_{\rm U} r_1^{-\alpha_v}}{\sigma^2 + \mathcal{I}_1+\mathcal{I}_2}>\gamma_0) \\ &~~~~~~~~ \quad +\mathcal{A}_{v}\mathcal{A}_{v}^2 \mathbb{P}(\frac{\sum_{k=1}^{K}P_{\rm a} G_{\rm a} h_{\rm a_0}^{(k)} r_2^{-\alpha_{\rm a}}}{\sigma^2 + \mathcal{I}_1+\mathcal{I}_2}>\gamma_0), \nonumber
\end{align}
where the main interference is from non-typical UAV $\mathcal{I}_1$  and all APs $\mathcal{I}_2$ for the typical receiver associated with UAV.
The Laplace transform of interference from non-typical UAV $ \mathcal{I}_1$ is derived as \eqref{I1} [\cite{article009090}, Eq.(A.4)],
\begin{figure*}[!t]
\rule{\textwidth}{0.4pt}
\begin{small}
\begin{align} \label{I1}
\begin{array}{l}
{{\cal L}_{{\mathcal{I}_1}}}(s) = {{\rm{E}}_{{\Phi _u}}}\left[ {{\prod _{{x_i} \in {\Phi _u}\backslash {x_o}}}\left( {{p_{\rm{L}}}({r_{{x_i}}}){{\rm{E}}_{{h_{\rm{U}}}}}\left[ {{e^{ - s{P_{\rm{U}}}{G_{\rm{U}}}{h_{\rm{U}}}r_{{x_i}}^{ - {\alpha _{\rm{L}}}}}}} \right] + {p_{{\rm{NL}}}}({r_{{x_i}}}){{\rm{E}}_{{h_{\rm{U}}}}}\left[ {{e^{ - s{P_{\rm{U}}}{G_{\rm{U}}}{h_{\rm{U}}}r_{{x_i}}^{ - {\alpha _{{\rm{NL}}}}}}}} \right]} \right)} \right]\\
 = \exp \{  - {\lambda _{\rm{U}}}^{ - 1} \times \\
\{ \left[ {\int\limits_0^{2\pi } {\int\limits_{{v_{\min }^{(1)}}(\theta )}^{{v_{\max }^{(1)}}(\theta )} {1 - {{(1 + \frac{{s{P_{\rm{U}}}{G_{\rm{U}}}{l ^{ - {\alpha _{\rm{L}}}}}}}{{{m_{\rm{L}}}}})}^{ - {m_{\rm{L}}}}}v{p_{\rm{L}}}({\rm{y}}){\zeta ^{(2)}}(v){\rm{d}}v{\rm{d}}\vartheta  + \int\limits_0^{2\pi } {\int\limits_{{v_{\min }^{(2)}}(\theta )}^\infty  {1 - {{(1 + \frac{{s{P_{\rm{U}}}{G_{\rm{U}}}{l ^{ - {\alpha _{\rm{L}}}}}}}{{{m_{\rm{L}}}}})}^{ - {m_{\rm{L}}}}}v{p_{\rm{L}}}({\rm{y}})\lambda _U^2{\rm{d}}v{\rm{d}}\vartheta } } } } } \right]\\
 + \left[ {\int\limits_0^{2\pi } {\int\limits_{{v_{\min }^{(1)}}(\theta )}^{{v_{\max }^{(1)}}(\theta )} {1 - {{(1 + \frac{{s{P_{\rm{U}}}{G_{\rm{U}}}{l ^{ - {\alpha _{{\rm{NL}}}}}}}}{{{m_{{\rm{NL}}}}}})}^{ - {m_{\rm{L}}}}}v{p_{{\rm{NL}}}}({\rm{y}}){\zeta ^{(2)}}(v){\rm{d}}v{\rm{d}}\vartheta  + \int\limits_0^{2\pi } {\int\limits_{{v_{\min }^{(2)}}(\theta )}^\infty  {1 - {{(1 + \frac{{s{P_{\rm{U}}}{G_{\rm{U}}}{l ^{ - {\alpha _{{\rm{NL}}}}}}}}{{{m_{{\rm{NL}}}}}})}^{ - {m_{{\rm{NL}}}}}}v{p_{{\rm{NL}}}}({\rm{y}})\lambda _U^2{\rm{d}}v{\rm{d}}\vartheta } } } } } \right]\},
\end{array}
\end{align}
\end{small}
where $ {v_{\min }^{(1)}(\vartheta )} = \max[d,2x\cos\vartheta] $, ${v_{\max }^{(1)}(\vartheta )} =\max[2d,2x\cos\vartheta] $, ${v_{\min }^{(2)}(\vartheta )} =\max[2d,2x\cos\vartheta]  $ and ${v_{\max }^{(2)}(\vartheta )} = \infty$, $x=\sqrt{r^2-H_{\rm U}^2}$, $l=\sqrt{r^2+v^2-2v\cos \vartheta \sqrt{r^2-H_{\rm U}^2}}$. \\ 
\rule{\textwidth}{0.4pt}
\begin{small}
\begin{align} \label{I1_1}
\mathcal{L}_{\mathcal{I}_1}(s) &= \mathbb{E}_{\Phi_u} \left[ \prod_{x_i \in \Phi_u \backslash x_o}
\left(
p_{\mathrm{L}}(r_{x_i}) \, \mathbb{E}_{h_{\rm U}}\left[ e^{-s P_{\rm U} G_{\rm U} h_{\rm U} r_{x_i}^{-\alpha_{\rm L}}} \right] 
+ p_{\mathrm{NL}}(r_{x_i}) \, \mathbb{E}_{h_{\rm U}} \left[ e^{-s P_{\rm U} G_{\rm U} h_{\rm U} r_{x_i}^{-\alpha_{\rm NL}}} \right]
\right) \right]\\
& =\exp\Biggl(
-\,2\pi\,\lambda
\int_{r}^{\infty}
\Bigl[
1 
-\,\Bigl(
    p_{\mathrm{L}}(r)\,\bigl(1 + \tfrac{s\,P_{\rm U}\,G_{\rm U} r^{-\alpha_{\rm L}}}{m_1}\bigr)^{-m_1}
    \;+\;
    p_{\mathrm{NL}}(r)\,\bigl(1 + \tfrac{s\,P_{\rm U}\,G_{\rm U} r^{-\alpha_{\rm NL}}}{m_1}\bigr)^{-m_1}
  \Bigr)
\Bigr]\,r\,dr
\Biggr).\notag
\end{align}
\end{small}
\rule{\textwidth}{0.4pt}
\begin{align} \label{Pcov}
\mathbb {P}_{\rm C} &= \mathbb{P}(\text{SINR} > \gamma_0)\\
&=\mathcal{A}_{v} \biggl\{\sum_{v \in\{\mathrm{L},\,\mathrm{NL}\}}\mathcal{A}_{v}^1\int_h^\infty \sum_{k=0}^{m - 1} \frac{(-s^\square)^k}{k!}
\frac{\mathrm{d}^k}{\mathrm{d}{s^\square}^k} \left[
e^{-s' \sigma^2} \mathcal{L}_{\mathcal{I}}(s^\square)
\right] f_R(r_1) \, \mathrm{d}r_1  + \mathcal{A}_{v}^2 ~\sum \limits _{n = 1}^{{k_{\text {S}}}} {{(- 1)^{n + 1}}\left ({{\begin{array}{*{20}{c}} {{k_{\text {S}}}} \\ n \end{array}} }\right)} {e^{ - {s}{\sigma ^{2}}}} {\mathcal {L}_{\mathcal {I}}}\left ({{{s}} }\right) \biggr \}, \nonumber\\
&\text{where $s^\square = \frac{m \gamma_0}{P_{\rm U}G_{\rm U}} r_1^{\alpha_v}$, $s = \frac{n\eta \gamma_0}{CG_{\rm a}}$, $\eta = \frac{\bigl(k_S!\bigr)^{-1/k_S}}{\theta_S}$. } \nonumber
\end{align}
\rule{\textwidth}{0.4pt}
\end{figure*}
\(\mathcal{I}_{2}\) is defined as 
\begin{equation}
{{\cal I}_2} = \underbrace {\sum\limits_{k = 1}^K {\sum\limits_{{x_i} \in {\Xi _{{L_0}}}} {\frac{{{P_{\rm{a}}}{G_{\rm{a}}}{h_{{{\rm{a}}_{\rm{0}}}}}}}{{r_{{x_i}k}^{{\alpha _{\rm{a}}}}}}} } }_{{{\cal I}_{21}}} + \underbrace {\sum\limits_{k = 1}^K {\sum\limits_{{x_i} \in {\Xi _{{L_i}}} \setminus {\Xi _{{L_0}}}} {\frac{{{P_{\rm{a}}}{G_{\rm{a}}}{h_{{{\rm{a}}_{\rm{1}}}}}}}{{r_{{x_i}k}^{{\alpha _{\rm{a}}}}}}} } }_{{{\cal I}_{22}}}
\end{equation}
where, $\mathcal{I}_{21}$ is the interference typical line APs in all layers, $\mathcal{I}_{22}$ is the interference from non-typical line in all layers.
The Laplace transform of the total interference $\mathcal{I}_2$ simplifies to
\begin{small}
\begin{align}
{{\cal L}_{{\mathcal{I}_2}}}(s) &= \exp \left\{ {  - 2\pi  \cdot \pi {\lambda _l}  } \right.\\
&\quad \times \left. {\sum\limits_{k=1}^K {{\lambda _{{a_k}}}} {\int_0^\infty } \left( {1 - {{\left( {1 + \frac{{s{P_{\rm a} G_{\rm a}}}}{m_2}{r^{ - {\alpha_{\rm a}}}}} \right)}^{ - m_2}}} \right)r{\mkern 1mu} dr} \right\}. \nonumber
\end{align}
\end{small}

For the typical receiver associated with AP, the interference $\mathcal{I}_1$ from all UAV and $\mathcal{I}_2$ from all APs on non-typical lines.
The Laplace transform of \(\mathcal{I}_{1}\) can be expressed as ~\eqref{I1_1}. The Laplace transform of \(\mathcal{I}_{2}\) can be expressed as
\begin{align}
{{\cal L}_{{{\cal I}_2}}}(s) = \prod\limits_{k = 1}^K {\exp } \left\{ {\left. \begin{array}{l}
 - 2\pi {\lambda _l}\int_0^\infty  [ 1 - \exp ( - 2{\lambda _a} \times \\
\int_0^\infty  {1 - \frac{1}{{{{(1 + \theta (x,y))}^{{m_2}}}}}dx} )]dy
\end{array} \right\}} \right.,
\end{align}
where $\theta (x,y) = \frac{{s{P_{\rm{a}}}{G_{\rm{a}}}}}{{{m_2}{{({x^2} + {y^2})}^{{\alpha _{\rm{a}}}/2}}}}$.

Thus, the coverage probability \eqref{pc2} can be derived as $\mathbb {P}_{\rm C}$ in (\ref{Pcov}), where $s' = \frac{m \gamma_0}{P_{\rm U}} r_1^{\alpha_1}$.
\begin{proof}
See Appendix \ref{appendix D}. 
\end{proof}

\section{Numerical Results}
We consider UAV distributed as a homogeneous PPP with density $\lambda_{\rm U} = 10/\mathrm{km}^2$ at a fixed altitude of $100,\mathrm{m}$. Roads follow a PLP with line density $\mu_l = 10/\mathrm{km}$, and APs are deployed along each road as a 1-D PPP with density $\lambda_{\rm a} = 2/\mathrm{km}$ across $K = 3$ layers.
For simplicity, UAV and APs are assumed to have identical characteristics: transmit powers $P_{\rm U} = 30~\mathrm{dBm}$, $P_{\rm a} = 23~\mathrm{dBm}$; path loss exponents $\alpha_{\rm L} = 2$, $\alpha_{\rm NL} = 3$, and $\alpha_{\rm a} = 3$; LoS model parameters $a = 12.08$, $b = 0.11$; mainlobe antenna gains $G_{\rm U} = G_{\rm a} = 1$ (linear scale); selection biases $B_{\rm U} = B_{\rm a} = 0\mathrm{dB}$; and Nakagami fading parameters $m_{\rm L} =3$, $m_{\rm NL} =3$, $m_{21} =  m_{22} = 3$, power scaling constant $C = 10^{-7}$.

\begin{figure}[!t]
\begin{minipage}[t]{0.495\linewidth}
    \centering
    \includegraphics[width=4.35cm, height=4cm]{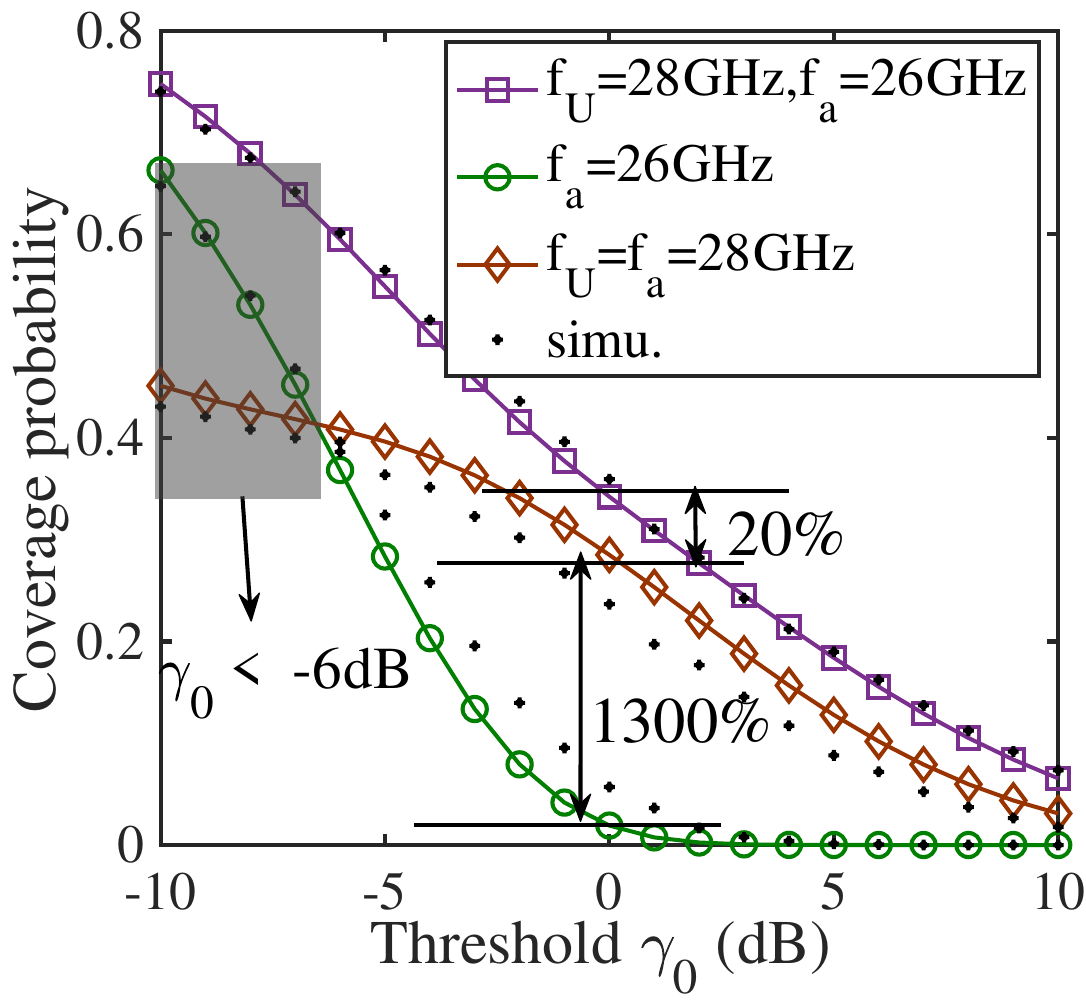}
    \caption{Coverage probability versus SNR threshold.}
    \label{cov}
\end{minipage}%
\hfill
\begin{minipage}[t]{0.495\linewidth}
    \centering
    \includegraphics[width=4.35cm, height=4cm]{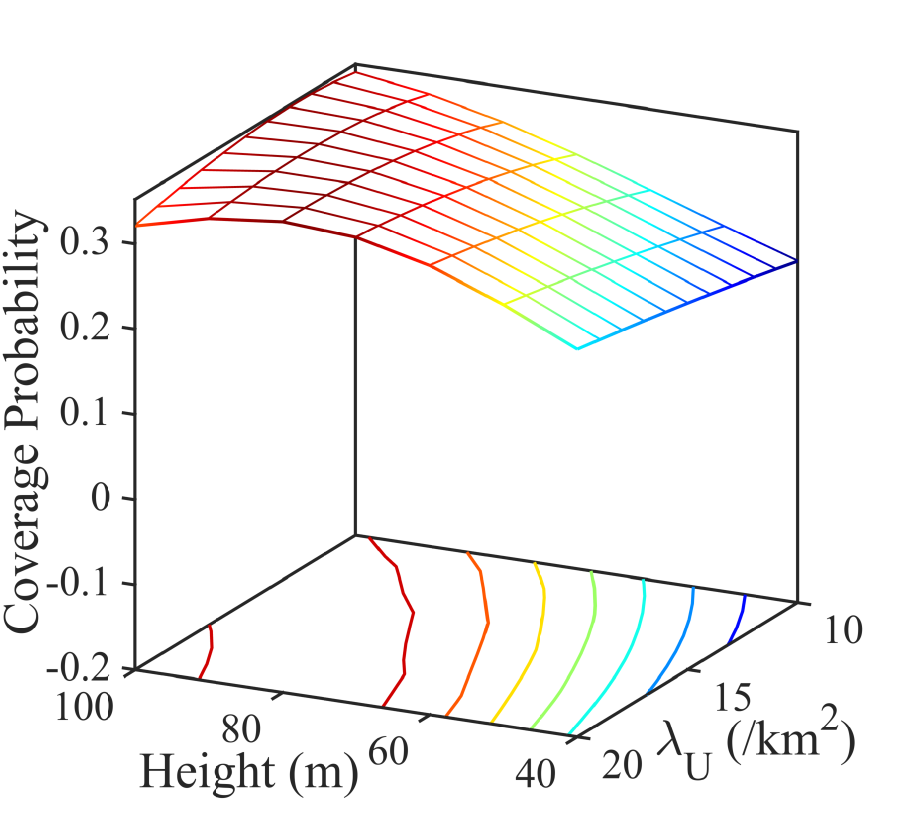}
    \caption{Coverage probability versus UAV density. ($\gamma_0 =0 \rm{dB}$)}
    \label{density}
\end{minipage}
\end{figure}

Fig.~\ref{cov} shows the coverage probability with the change of SINR threshold for 3 scenarios, pure AP, adding UAV, and adding UAV with different frequency bands, resepectively. The simulation results is almost same with theoretical results. It is observed that adding UAV with different frequencies achieves the highest coverage probability, about 20\% absolute improvement than with same frequency. For low-SNR regime ($\gamma_0 < -6 \rm{dB}$), pure AP coverage is much better than with UAV, because of the additional interference from UAV.

Fig.~\ref{density} illustrates the coverage probability as a function of UAV density and UAV altitude.  
It is observed that increasing the UAV density improves coverage probability due to a higher chance of favorable LoS links.  
However, the improvement becomes less significant when the UAV density is already high, showing a saturation effect.  
Moreover, a moderate UAV altitude provides better performance by balancing the trade-off between path loss and LoS probability.

\begin{figure}[!t]
\begin{minipage}[t]{0.495\linewidth}
    \centering
    \includegraphics[width=4.4cm, height=4cm]{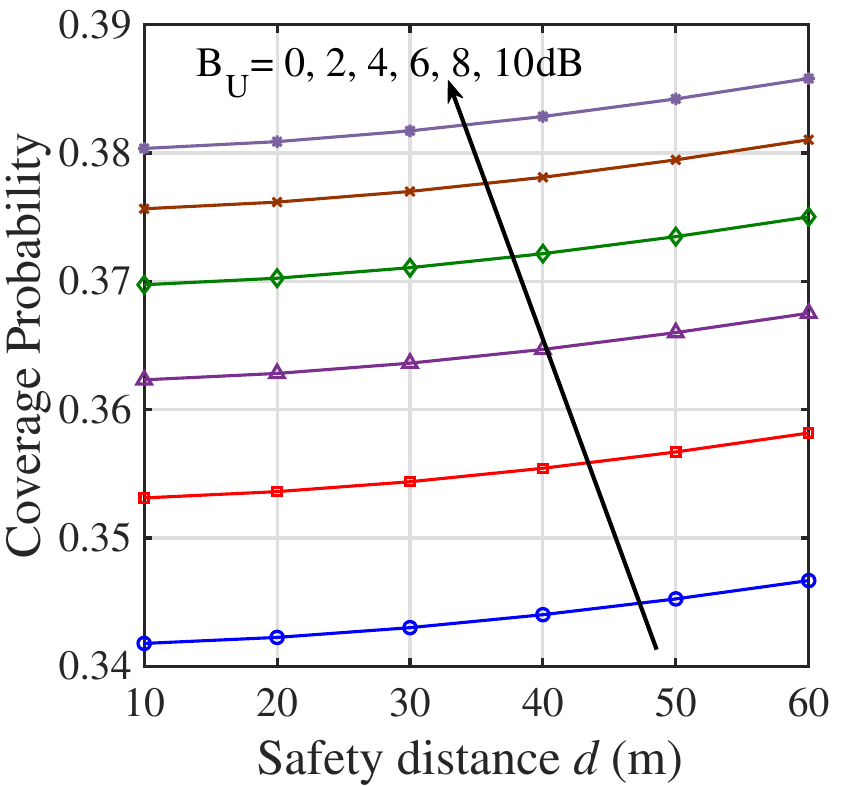}
    \caption{Coverage probability versus safety distance $d$ with different bias $B_{\rm U}$ in $\rm dB$ ($\gamma_0 =0 \rm{dB}$).}
    \label{distance}
\end{minipage}%
\hfill
\begin{minipage}[t]{0.495\linewidth}
    \centering
    \includegraphics[width=4.4cm, height=4cm]{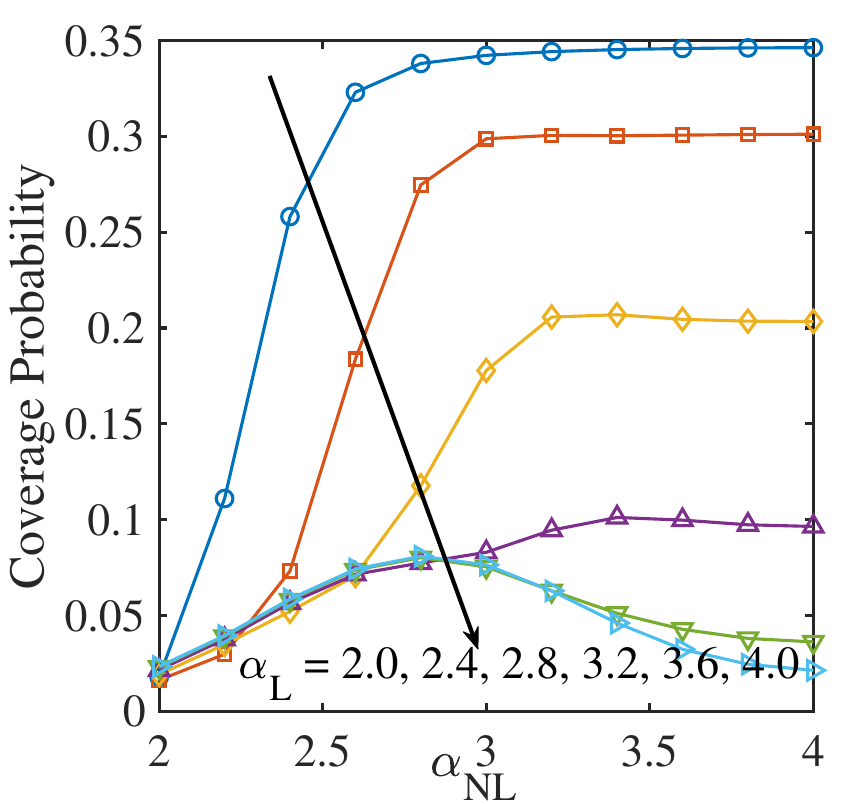}
    \caption{Coverage probability versus LoS and NLoS path loss exponent ($\gamma_0 =0 \rm{dB}$).}
    \label{pathloss}
\end{minipage}
\end{figure}

Fig.~\ref{distance} depicts the coverage probability as a function of the safety distance $d$ for various UAV selection biases $B_{\rm U}$.  
It can be observed that increasing  $B_{\rm U}$ improves the coverage performance, as a larger bias encourages users to associate with UAV, which typically provide better LoS conditions.  
In addition, a larger safety distance slightly enhances the coverage probability by mitigating strong interference from nearby UAV.  

Fig.~\ref{pathloss} illustrates the coverage probability versus the NLoS path loss exponent \(\alpha_{\rm NL}\), with different values of the LoS path loss exponent \(\alpha_{\rm L}\).  
It is observed that the coverage probability increases with \(\alpha_{\rm NL}\), as higher attenuation in NLoS links reduces interference from distant APs.  
Conversely, a lower \(\alpha_{\rm L}\) leads to better coverage due to less attenuation in LoS channels.  
The best performance is achieved when \(\alpha_{\rm L}\) is small and \(\alpha_{\rm NL}\) is large, highlighting the importance of strong LoS connectivity and suppressed NLoS interference.

\section{Conclusion}
This paper analyzed the performance of UAV-assisted vehicular cell-free networks using a stochastic geometry framework that accounts for practical road-constrained deployments and realistic wireless propagation.  
Roads were modeled by PLP, with vehicles and access points distributed along the roads via independent 1-D PPPs.  
To manage the analytical complexity introduced by signal aggregation from multiple APs, a distance-based power control strategy was proposed.
Simulation results validated the theoretical analysis and revealed the impact of key system parameters on the coverage probability.

\section*{Acknowledgment}
This work is financed by the European Commission under the Horizon Europe MSCA programme (HORIZON-MSCA-2024-SE-01-01), Grant Agreement No. 101236523 (AeroNet project).

\appendix
\subsection{Proof of Lemma 1} \label{appendixB}
To evaluate the association probability of a typical receiver with a UAV, we consider the condition under which the received power from the UAV exceeds that from the strongest AP. Specifically, we derive the probability as
\begin{small}
\begin{align}
&\mathbb{P}(P_{\rm U}B_{\rm U}G_{\rm U}R_{1}^{-\alpha_v} > P_{\rm a}B_{\rm a}G_{\rm a}R_{2}^{-\alpha_{\rm a}})
= \mathbb{P}(R_{1} < \xi_{21}^{-\frac{1}{\alpha_1}} R_{2}^{\frac{\alpha_{\rm a}}{\alpha_v}})\nonumber\\
&\overset{(a)}{=} \mathbb{E}_{R_{2}}\left[\mathbb{P}(R_{1} < \xi_{21}^{-\frac{1}{\alpha_v}} R_{2}^{\frac{\alpha_{\rm a}}{\alpha_v}} \mid R_{2})\right] \\
&= \int_{0}^{\infty} F_{R_{1}}(\xi_{21}^{-\frac{1}{\alpha_v}} r_{2}^{\frac{\alpha_{\rm a}}{\alpha_v}}) f_{R_{2}}(r_{2})\mathrm{d}r_{2} \nonumber\\
&\overset{(b)}{=} \int_0^\infty  \left[1 - {{e^{ - {\lambda _u}\pi \rho(r_2)}}}\right] 2{\lambda_{\rm a}}{e^{ - 2{\lambda_{\rm a}}{r_2}}} \mathrm{d}r_{2} \nonumber\\
&\overset{(c)}{=} 1 - \int_0^\infty  {{e^{ - {\lambda _u}\pi \rho(r_2)}}} 2{\lambda _a}{e^{ - 2{\lambda _a}{r_2}}} \mathrm{d}r_{2}, \nonumber
\end{align}
\end{small}
where $\rho(r_2) = \max(0,{\xi_{21} ^{ - \frac{2}{{{\alpha_v}}}}}r_2^{\frac{{2{\alpha_{\rm a}}}}{{{\alpha _v}}}}-{H_{\rm U}}^2) $, $\xi_{21} = \frac{P_{\rm a}B_{\rm a}G_{\rm a}}{P_{\rm U}B_{\rm U}G_{\rm U}}$, $v\in\{\mathrm{L,NL}\}$, $(a)$ is rewritten using conditional probability, where we condition on the value of \(R_{2}\), $(b)$ follows from cumulative distribution \(F_{R_{1}}(x)\) and the probability density \(f_{R_{2}}(x)\), $(c)$ follows from integral split, then proof is completed.

\subsection{Proof of Lemma 2} \label{appendixC}
We now derive the conditional CDF of the distance from the typical receiver to its serving UAV, conditioned on the association event $\mathcal{A}_v^1$. This conditional CDF is denoted as $F_{R_1}(r \mid \mathcal{A}_v^1)$, and is given by
\begin{align}
F_{R}(r|\mathcal{A}_{v}^1 ) =& \mathbb{P}(R_1 \le r \mid \mathcal{A}_{v}^1 )  \overset{(a)}{=} 1 - \frac{\mathbb{P}(R_1 > r, \mathcal{A}_{v}^1 )}{\mathbb{P}(\mathcal{A}_{v}^1 )} \\\overset{(b)}{=}& 1 - \frac{1}{\mathbb{P}(\mathcal{A}_{v}^1 )} \mathbb{P}\left( r < R_1 < \xi_{21}^{-1/\alpha_v} R_2^{\alpha_{\rm a}/\alpha_v} \right) \nonumber\\\overset{(c)}{=}& 1 - \frac{1}{{{\mathbb{P}}(\mathcal{A}_{v}^1 )}} \int_{\xi_{21}^{1/{\alpha_v}}{r^{{\alpha _v}/{\alpha_{\rm a}}}}}^\infty  {\left\{ {\left[ {{F_{{R_1}}}\left( {\xi_{21}^{ - 1/{\alpha_v}}r_2^{{\alpha_{\rm a}}/{\alpha _v}}} \right)} \right.} \right.} \nonumber\\ &~~~~~~~~~~~~~~~~~~~~~~~~~~~~~~~\left. {- {F_{{R_1}}}(r)] \times {f_{{R_2}}}({r_2})} \right\}, \nonumber\end{align}
where $(a)$ follows from the definition of conditional probability, $(b)$ incorporates the event condition $\mathcal{A}_v^1$, and $(c)$ applies the law of total probability over $R_2$.
Accordingly, the conditional PDF of $R_2$ follows from the same derivation steps as
\begin{align}
F_{R}(r \mid \mathcal{A}_{v}^2 )
&= 1 - \frac{1}{\mathbb{P}(\mathcal{A}_{v}^2 )}
\int_{\xi_{12}^{1/\alpha_{\rm a}} r^{\alpha_{\rm a}/\alpha_v}}^{\infty}  \\
&\quad \left[
F_{R_2}\left( \xi_{12}^{-1/\alpha_{\rm a}} r_1^{\alpha_v / \alpha_{\rm a}} \right) - F_{R_2}(r)
\right]
f_{R_1}(r_1)\, dr_1. \notag
\end{align}

Differentiating with respect to $r$, we get the conditional PDF
\begin{align}	
{f_{R}}(r|{\mathcal{A}_{v}^1 }) 
&= \frac{d}{{dr}}{F_{{R_1}}}(r|{\mathcal{A}_{v}^1 }) \\[-4pt]
&= \frac{1}{{{\mathbb P}({\mathcal{A}_{v}^1 })}}{f_{{R_1}}}(r)\,
   {\mathbb P}\!\left( {{R_2} > \xi _{21}^{1/{\alpha _v}}{r^{{\alpha _v}/{\alpha_{\rm a}}}}} \right). \nonumber
\end{align}
Having $\mathbb{P}(R_2 > x) = \exp(-2 \lambda_{\rm a} x),$
Thus, the conditional PDF becomes:
\begin{align}
	&{f_{R}}(r|{\mathcal{A}_{v}^1}) = \frac{{2\pi {\lambda _u}r}}{{{\mathbb{P}}(\mathcal{A}_{v}^1 )}}{e^{ - \pi {\lambda _u}{(r^2-{H_{\rm U}}^2)} - 2{\lambda _a}\xi _{21}^{1/{\alpha _v}}{r^{{\alpha _v}/{\alpha_{\rm a}}}}}},  \\
   & f_{R}(r|\mathcal{A}_{v}^2 )
= \frac{2\lambda_{\rm a}}{\mathbb{P}(\mathcal{A}_{v}^2 )} e^{-2\lambda_{\rm a} r -\lambda_u \pi \rho(r)}, 
\end{align}
where $\rho (r) = \max (0,\xi _{12}^{\frac{2}{{{\alpha _{\rm{a}}}}}}{r^{\frac{{2{\alpha _{\rm{a}}}}}{{{\alpha _v}}}}} - {H_{\rm{U}}}^2)$, $\xi_{21} = \frac{P_{\rm a}B_{\rm a}G_{\rm a}}{P_{\rm U}B_{\rm U}G_{\rm U}}$, $\xi_{12} = \frac{P_{\rm U} B_{\rm U} G_{\rm U}}{P_{\rm a} B_{\rm a} G_{\rm a}}$.

\subsection{Proof of Coverage Probability}
\label{appendix D}
We first consider the event $\mathcal{E}_1$, where the typical receiver is served by a UAV from tier-1. The coverage probability under this event can be expressed as
\begin{align}\label{coverage e1}
\mathbb{P}_{\rm C}
&= \mathbb{P} \left( \frac{P_{\rm U} G_{\rm U} h_{\rm U} r_1^{-\alpha_1}}{\sigma^2 + \mathcal{I}} > \gamma_0 \right) \\
&=\int_{H_{\rm U}}^\infty \sum_{k=0}^{m - 1} \frac{(-s^\square)^k}{k!}
\cdot \frac{\mathrm{d}^k}{\mathrm{d}{s^\square}^k} \left[
e^{-s^\square \sigma^2} \cdot \mathcal{L}_{\mathcal{I}}(s^\square)
\right] f_R(r_1) \, \mathrm{d}r_1, \nonumber
\end{align}
where $s^\square = \frac{m \gamma_0}{P_{\rm U} G_{\rm U}} r_1^{\alpha_1}$.
\textit{Proof.}
Following the same approach presented in [\cite{8340239}, Eq.(32)], so skipped,

Our goal is to evaluate the coverage probability $\mathbb{P}(\gamma > \gamma_0)$. To this end, we first characterize the distribution of the total received signal power from the nearest APs in each layer. Since the exact distribution is analytically intractable, we adopt a second-order moment matching to approximate the aggregated signal power by a Gamma distribution. Since the typical receiver is served by $K$ cooperating APs, the total received signal power can be written as
$C G_{\rm a} \sum_{i=1}^{K} h_{\rm a_0}^{(k)}.$

If $S = \sum_{i=1}^{K} h_{\rm a_0}^{(k)}$, where each $h_{\rm a_0}$ follows a Gamma distribution $h_{\rm a_0} \sim \Gamma(k_i, \theta_i)$, then the expectation of $S$ is given by $\mathbb{E}[S] = \sum_{i=1}^{K} k_i \theta_i$, and the variance of $S$ is given by $\mathrm{Var}(S) = \sum_{i=1}^{K} k_i \theta_i^2$.

We approximate \( S \) by a Gamma distribution $\Gamma(k_S, \theta_S)$\cite{5953530}, where $k_S = \frac{\left( \mathbb{E}[S] \right)^2}{\text{Var}(S)}$ and $\theta_S = \frac{\text{Var}(S)}{\mathbb{E}[S]}$, substituting the mean and variance expressions as $k_S = \frac{\left( \sum_{i=1}^{K} k_i \theta_i \right)^2}{\sum_{i=1}^{M} k_i \theta_i^2},$ and $\theta_S = \frac{\sum_{i=1}^{K} k_i \theta_i^2}{\sum_{i=1}^{M} k_i \theta_i}.$

The coverage probability becomes
\begin{equation}
\setlength{\abovedisplayskip}{4pt}
\setlength{\belowdisplayskip}{4pt}
	\mathbb{P}\left( \frac{C G_{\rm a}\sum_{i=1}^{K} h_{\rm a_0}}{\mathcal{I} + \sigma^2} > \gamma_0 \right)
= \mathbb{P}\left( S > \mu (\mathcal{I} + \sigma^2) \right), 
\end{equation}
where $\mu = \frac{\gamma_0}{C G_{\rm a}}$, \( S \sim \Gamma(k_S, \theta_S) \).
Assuming $k_S$ is an integer, we can obtain
\begin{align}
\mathbb{P}_{\rm C}(\gamma_0) &= \mathbb{P}\left[h_{\rm a_0} > \frac{\gamma_0}{C G_{\rm a}} \left( \mathcal{I}_1 + \mathcal{I}_2 + \sigma^2 \right) \right] \\
&\overset{(a)}{\approx} 1 - \mathbb{E}_{\Phi_{c}} \left[ \left( 1 - e^{- \frac{\eta \gamma_0}{CG_{\rm a}} (\mathcal{I} + \sigma^2)} \right)^{k_{\text{S}}} \right] \nonumber \\
&= \sum_{n=1}^{k_{\text{S}}} (-1)^{n+1} \binom{k_{\text{S}}}{n} \, 
\mathbb{E} \left[ e^{ - \frac{n \eta \gamma_0}{CG_{\rm a}} (\mathcal{I} + \sigma^2) } \right]\nonumber \\
&= \sum_{n=1}^{k_{\text{S}}} (-1)^{n+1} \binom{k_{\text{S}}}{n} \, 
e^{- s \sigma^2} \, \mathcal{L}_{\mathcal{I}}(s), \nonumber
\end{align}
where $s = \frac{n\eta \gamma_0}{CG_{\rm a}}$, $\eta = \frac{\bigl(k_S!\bigr)^{-1/k_S}}{\theta_S}$, and we have used the assumption that $k_S$ is integer, $(a)$ comes from Appendix A.2 of \cite{article009090}.  Then proof is completed.

\bibliographystyle{IEEEtran}
\bibliography{mybib}
\end{document}